%% file: main.tex
\documentclass[conference]{IEEEtran}
\usepackage{cite}
\usepackage{amsmath,amssymb,amsfonts}
\usepackage{algorithmicx}
\usepackage{graphicx}
\usepackage{caption}
\usepackage{subcaption}
\usepackage{textcomp}
\usepackage[dvipsnames]{xcolor}
\usepackage{orcidlink}
\usepackage{hyperref}

\usepackage{tikz}
\usetikzlibrary{quantikz2}

\usepackage[english]{babel}
\usepackage{amsthm}
\theoremstyle{definition}
\newtheorem{definition}{Definition}[section]
\newtheorem{theorem}{Theorem}[section]
\newtheorem{lemma}{lemma}[section]

\newtheorem*{remark}{Remark}
\newtheorem{example}{Example}[section]
\usepackage{algorithm2e}
\RestyleAlgo{ruled}
\usepackage[noend]{algpseudocode}

\usepackage{listings}

\usepackage[capitalise]{cleveref}

\def\BibTeX{{\rm B\kern-.05em{\sc i\kern-.025em b}\kern-.08em
    T\kern-.1667em\lower.7ex\hbox{E}\kern-.125emX}}
\begin{document}

\title{Reducing Mid-Circuit Measurements via Probabilistic Circuits\\
{
}
\thanks{The research is part of the Munich Quantum Valley (MQV), which is supported by the Bavarian state government with funds from the Hightech Agenda Bayern Plus. We are grateful to our supervisor Helmut Seidl for many fruitful discussions and his support at all times.}
}

\author{\IEEEauthorblockN{1\textsuperscript{st} Yanbin Chen \orcidlink{0000-0002-1123-1432}}
\IEEEauthorblockA{\textit{TUM School of CIT} \\
\textit{Technical University of Munich}\\
Germany \\
yanbin.chen@tum.de }
\and
\IEEEauthorblockN{2\textsuperscript{nd} Innocenzo Fulginiti \orcidlink{0000-0001-8818-9626}}
\IEEEauthorblockA{\textit{TUM School of CIT} \\
\textit{Technical University of Munich}\\
Germany \\
innocenzo.fulginiti@tum.de}
\and
\IEEEauthorblockN{3\textsuperscript{rd} Christian B.~Mendl \orcidlink{0000-0002-6386-0230}}
\IEEEauthorblockA{\textit{TUM School of CIT} \\
\textit{Technical University of Munich}\\
Germany \\
christian.mendl@tum.de}
}

\newcommand{\andothers}{et~al.}

\newcommand{\probgate}[2]{\gate{#1}\gategroup[1,steps=1,style={rounded
         corners,fill=gray!20, inner
         xsep=2pt},background,label style={label
         position=below,anchor=north,yshift=-0.2cm}]{#2}}

\newcommand{\probctrl}[3]{\ctrl{#1}\gategroup[#2,steps=1,style={rounded
         corners,fill=gray!20, inner
         xsep=2pt},background,label style={label
         position=below,anchor=north,yshift=-0.2cm}]{#3}}

\maketitle

\begin{abstract}
\input{sections/abstract}
\end{abstract}


\section{Introduction}\label{sec:intro}
\input{sections/introduction}

\section{Preliminaries}\label{sec:prelim}
\input{sections/preliminaries}

\section{Method}\label{sec:methods}
\input{sections/method}

\section{Related works}\label{sec:rel_works}
\input{sections/relevant_work}

\section{Conclusion}\label{sec:conclusion}
\input{sections/conclusion}

\section*{Acknowledgment}
The research is part of the Munich Quantum Valley (MQV), which is supported by the Bavarian state government with funds from the Hightech Agenda Bayern Plus. We are grateful to Prof. Dr. Helmut Seidl for many fruitful discussions and his support at all times.

\end{document}

%% file: sections/abstract.tex
Mid-circuit measurements and measurement-controlled gates are supported by an increasing number of quantum hardware platforms and will become more relevant as an essential building block for quantum error correction. However, mid-circuit measurements impose significant demands on the quantum hardware due to the required signal analysis and classical feedback loop. This work presents a static circuit optimization algorithm that can substitute some of these measurements with an equivalent circuit with randomized gate applications. Our method uses ideas from constant propagation to classically precompute measurement outcome probabilities. 
Our proposed optimization is efficient, as its runtime scales polynomially on the number of qubits and gates of the circuit.

%% file: sections/introduction.tex
Quantum dynamic circuits are quantum circuits with mid-circuit measurements.
Such circuits offer a flexible framework for realizing quantum algorithms, where the gate sequence is not fully determined at compile time since operations can depend on measurement outcomes.
Mid-circuit measurements enable, e.g.,
employing a qubit multiple times: 
By measuring a qubit, it is reset and, therefore, ready to be reused for further computation as ``fresh'' qubit \cite{decross_qubit-reuse_2022, brandhofer_optimal_2023, hua_exploiting_2023}.

Mid-circuit measurements are nowadays supported by quantum software frameworks like Qiskit, Pennylane or T$\ket{\text{ket}}$\cite{ibm_mid_circ_meas_available_2021, nation_ibm_howtomidcircmeas_2021, ibm_dynamic_circuit, Qiskit, bergholm_2022_pennylane, Sivarajah_tket_2021}. 
Nevertheless, they impose large requirements on the quantum hardware and have a relatively long duration due to the required classical feedback loop \cite{PhysRevLett.127.100501, lubinski2022advancing, ella2023quantumclassical}.

\begin{figure}
\centering
\begin{subfigure}{0.4\textwidth}
   \begin{quantikz}[row sep={7mm,between origins}]
        \qw & \gate[2]{U_1}\slice[style=blue!60]{}  & \rstick[2]{$\ket{\Phi}$} & \qw & \targ{}  & \qw      & \ctrl{1} & \meter{} \wire[d][3]{c} &  \gate[4]{U_3} & \qw \\
        \qw & \qw                    & \qw                      & \qw & \ctrl{-1}& \ctrl{1} & \targ{}  & \qw                     &  \qw           & \qw \\
        \qw & \gate[2]{U_2}          & \rstick[2]{$\ket{\Psi}$} & \qw & \qw      & \targ{}  & \qw      & \qw                     &  \qw           & \qw \\
        \qw & \qw                    & \qw                      & \qw & \qw      & \qw      & \slice{} & \gate{V}                &  \qw           & \qw
    \end{quantikz}
    \caption{A circuit containing one mid-circuit measurement}
    \label{fig:example_mid_circ_elim_before}
\end{subfigure}
\hfill
\begin{subfigure}{0.4\textwidth}
   \begin{quantikz}[row sep={7mm,between origins}]
        \qw & \gate[2]{U_1}  &  & \qw & \targ{}  & \qw      & \ctrl{1} &   \qw & \qw & \gate[4]{U_3} & \qw \\
        \qw & \qw            &                       & \qw & \ctrl{-1}& \ctrl{1} & \targ{}  &   \qw & \qw &\qw & \qw \\
        \qw & \gate[2]{U_2}  &  & \qw & \qw      & \targ{}  & \qw      &   \qw & \qw & \qw & \qw  \\
        \qw & \qw            &                      & \qw & \qw      & \qw      &   \qw & \qw & \qw & \qw      &   \qw
    \end{quantikz}
    \caption{A measurement-free circuit equivalent to \cref{fig:example_mid_circ_elim_before}}
    \label{fig:example_mid_circ_elim_after}
\end{subfigure}
\caption{Example of mid-circuit measurement elimination, where $\ket{\Phi} = \frac{1}{\sqrt{2}}(\ket{00} + \ket{11})$ and $\ket{\Psi}$ is an arbitrary two-qubit state.}
\label{fig:example_mid_circ_elim}
\end{figure}
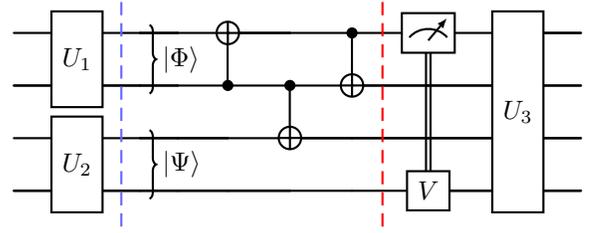
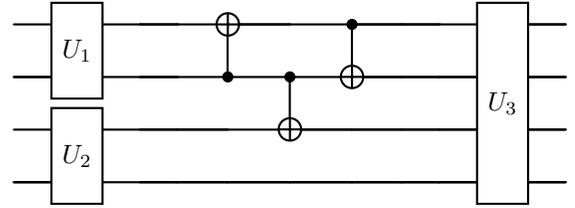
Some mid-circuit measurements, though, are redundant. 
\cref{ex:toy_mid_circ_elim} describes a toy model case.
\begin{example}
    The circuit in \cref{fig:example_mid_circ_elim_before} contains one mid-circuit measurement, with the measurement outcome controlling the gate $V$. After gates $U_1$ and $U_2$ are applied, we assume that the first two qubits are entangled in state $\ket{\Phi} = \frac{1}{\sqrt{2}}(\ket{00} + \ket{11})$ and the last two qubits are entangled in an arbitrary two-qubit state $\ket{\Psi}$.
    A short calculation shows that the state of the first qubit at the second dashed line is $\ket{0}$. Therefore, one can remove the mid-circuit measurement and its controlled gate $V$ since the measurement outcome will always be $0$, which leads to the optimized circuit in \cref{fig:example_mid_circ_elim_after}.
    \label{ex:toy_mid_circ_elim}
\end{example}

At the time of writing, state-of-the-art quantum compilers, such as Qiskit and t$\ket{\text{ket}}$ \cite{Qiskit, Sivarajah_tket_2021}, have not yet considered optimization passes for mid-circuit measurements. 
To fill this gap, we propose a solution to reduce the number of mid-circuit measurements in dynamic circuits in this paper.
Our solution requires that a mid-circuit measurement is performed on a pure state.
The measurement and the gate controlled by this measurement outcome could be replaced by a standard rotation gate followed by a probabilistic gate and a controlled gate.
The concepts of probabilistic gates and probabilistic circuits are introduced in the later sections.
Importantly, whether a probabilistic gate is applied or not is determined in each shot by a classical computer at compile time. 
Therefore, replacing dynamic components containing mid-circuit measurements with probabilistic components reduces the runtime overhead caused by mid-circuit measurements.

Our method extends Quantum Constant Propagation (QCP) to first perform a static analysis. It then applies a purity test which uses the constant information gathered from QCP to detect places where mid-circuit measurements could be replaced by probabilistic circuit snippets.

%% file: sections/preliminaries.tex
This paper assumes that readers have a basic understanding of quantum computing; for a more in-depth explanation of quantum computing, readers could refer to the textbook 
\cite{nielsen_QC_2012}.
In the following, we give a brief introduction to some concepts that will be used in the following sections of this paper. 

\paragraph{Compile time for quantum circuits}
In the context of this paper, compile time for quantum circuits is the period during which circuits are processed on the classical computer. In particular, circuits are optimized and tailored down to the target backend.
The ideology this paper uses to reduce the overhead of executing quantum circuits is to move some expensive runtime operations back to compile time if possible.  

\paragraph{Runtime for quantum circuits}
The runtime starts when the circuit is submitted to the target backend. At runtime, circuits are executed on a quantum device. Nowadays, quantum backends are supporting more and more operations, an important one of which is the mid-circuit measurement.

\paragraph{Mid-circuit measurement}
\begin{figure}
    \centering
    \begin{quantikz}[row sep={7mm,between origins}]
          \qw & \ \ldots\ & \gate{U_i} & 
          \meter{}
    & \gate{U_{i+1}}  & \ \ldots\ & \qw
     \end{quantikz}
    \caption{Illustration of a mid-circuit measurement.}
    \label{fig:mid_circ_meas}
\end{figure}
 
With support for mid-circuit measurements, it is possible to perform measurements at any place in the circuit.
As illustrated in \cref{fig:mid_circ_meas}, a measurement appears in the middle of the circuit rather than in the end.
In the simple example of mid-circuit measurement shown in \cref{fig:mid_circ_meas_control}, the outcome of the measurement is used to control the gate $U$, where $U$ is applied only if the measurement outcome is $1$.

\begin{figure}
    \centering
    \begin{quantikz}[row sep={10mm,between origins}]
          & \meter{} \wire[d][1]{c} & \qw \\
          & \gate{U} & \qw
    \end{quantikz}
    \caption{Example of a dynamic component in a quantum circuit.}
    \label{fig:mid_circ_meas_control}
\end{figure}

\paragraph{Static and dynamic circuit}
In the context of this paper, static circuits are quantum circuits that can be fully determined at compile time.
In dynamic circuits, on the other hand, some components require results of runtime operations, such as mid-circuit measurements, to be determined.
In general, a dynamic circuit contains static parts and dynamic parts, as illustrated in \cref{fig:ex_dynamic_circuit}.
The static part of the circuit is determined at compile time; whereas the dynamic part depends on the outcome of the mid-circuit measurement.

\begin{remark}
    The static part is fully determined at compile time. However, the dynamic part can only be completely determined at runtime. Sometimes, compile time is referred to as static time. 
\end{remark}

\begin{figure}
    \centering
    \begin{quantikz}[row sep={7mm,between origins}]
        \qw & 
         \gate[2]{U_1}    
         \gategroup[4,steps=5,style={dashed,rounded corners,fill=red!20, inner xsep=2pt},background,label style={label position=below,anchor=north,yshift=-0.2cm}]{static part}
         & \qw & \targ{}  & \qw      & \ctrl{1} & \meter{} \wire[d][3]{c} 
         \gategroup[4,steps=1,style={dashed,rounded corners,fill=blue!20, inner xsep=2pt},background,label style={label position=above,anchor=north,yshift=+0.3cm}]{dynamic part}
         &  \gate[4]{U_3} 
         \gategroup[4,steps=1,style={dashed,rounded corners,fill=red!20, inner xsep=2pt},background,label style={label position=below,anchor=north,yshift=-0.2cm}]{static part}
         & \qw \\
        \qw                    & \qw                      & \qw & \ctrl{-1}& \ctrl{1} & \targ{}  & \qw                     &  \qw           & \qw \\
        \qw & \gate[2]{U_2}       & \qw & \qw      & \targ{}  & \qw      & \qw                     &  \qw           & \qw \\
        \qw & \qw                                       & \qw & \qw      & \qw      & & \gate{V}                &  \qw           & \qw
    \end{quantikz}
    \caption{The dynamic circuit in \cref{ex:toy_mid_circ_elim}.}
    \label{fig:ex_dynamic_circuit}
\end{figure}

\paragraph{Quantum Constant Propagation}
Quantum Constant Propagation (QCP) is an optimization technique that efficiently simplifies controlled gates by propagating the initial constant information throughout the circuit \cite{chen_QCP_2023}.
The following optimization step performed by QCP gives a feeling of how this technique works.
\begin{align}
    \begin{quantikz}[column sep=5pt, row sep={20pt,between origins}]
        \lstick{\ket{0}} & \gate{H} & \ctrl{1} &\slice[style=black!60]{$\frac{1}{\sqrt{2}} (\ket{010} + \ket{100})$} & \ctrl{2} & \ctrl{1} &  & \\
        \lstick{\ket{0}} & & \targ{} & \gate{X} & \ctrl{1} & \targ{} & \ctrl{1}  & \\
        \lstick{\ket{0}} & & & & \targ{} & & \gate{V} &
    \end{quantikz}   
\equiv
\begin{quantikz}[column sep=5pt, row sep={20pt,between origins}]
        \lstick{\ket{0}} & \gate{H} & \ctrl{1} & & \ctrl{1} & \\
        \lstick{\ket{0}} & & \targ{} & \gate{X} & \targ{} &   \\
        \lstick{\ket{0}} & \gate{V} & & &  &
    \end{quantikz} \,.
\label{eq:qcp_example}
\end{align}
It can be observed that the state of the control qubits at the point indicated by the dashed line in the left-hand side of \cref{eq:qcp_example} is $\frac{1}{\sqrt{2}}(\ket{01} + \ket{10})$. This means that the controls are not satisfied; for this reason, the Toffoli gate is never executed and thus is removed.
On the other hand, the control of the gate $V$ will always be satisfied. Hence, QCP simplifies the controlled gate to a non-controlled one. 
However, propagating constant information is as hard as simulating circuits.
Therefore, QCP uses a dedicated data structure and performs a restricted simulation which has a polynomial complexity.  
The restricted simulation of QCP only tries to track constant information within a preset limit.  
Suppose $c$ is the maximum number of controls each gate may have, for a circuit of $g$ gates and $n$ qubits, 
the QCP will terminate in $\mathcal{O}(g \cdot c^2 \cdot n)$ steps.

%% file: sections/method.tex
In this section, our optimization of mid-circuit measurements is fully presented.

To begin with, the concept of circuit ensemble is given as follows, which will be used in the later discussion of this section to capture uncertainty about the circuit.
\paragraph{Uncertainty of quantum system}
Similar to ensembles of pure states, this paper uses ensembles to model uncertain quantum circuits.
\begin{definition} [The ensemble of circuits]
    Assume $n$ is a positive integer. If the unitary matrix for a circuit is $C_i$ with a probability of $p_i$ for each $i \in \{1, \dots, n\}$, then this uncertain quantum circuit can be described by the ensemble $\{(p_1, C_1), \dots, (p_n, C_n)\}$.
\end{definition}

\begin{example}
    Suppose there is a situation where a coin is tossed to decide which gate to perform: An $X$ gate will be performed for a head or an $I$ gate for a tail. 
    Then, before the coin is tossed, the operation to be performed can be described by the ensemble $\{(0.5, X), (0.5, I)\}$.
\end{example}
 
Next, the concepts of probabilistic gate and probabilistic circuit are introduced, which are essential to our method. 
\paragraph{Probabilistic circuit}
Probabilistic gates are non-deterministic quantum operations, of which the non-determinism comes from their stochastic compilation procedure, as described in \cref{def:prob_gate} and \cref{def:compile_prob_gate}.
\begin{definition}[Probabilistic quantum gate]
\label{def:prob_gate}
    A probabilistic quantum gate $U(p)$ consists of a quantum gate $U$ and a probability $p$. The circuit diagram of a probabilistic quantum gate $U(p)$ is shown in \cref{fig:example_prob_gate}.
\end{definition}
For a probabilistic gate $U(p)$, the probability $p$ is supposed to be viewed as a parameter fed to the compiler. At compile time, the gate $U(p)$ corresponds to a stochastic compilation procedure which generates a gate $U$ at a probability of $p$ or an identity gate otherwise, as described by \cref{def:compile_prob_gate}.

\begin{definition}[Compilation of probabilistic gate]
\label{def:compile_prob_gate}
    A probabilistic gate $U(p)$ compiles to a gate $U$ at a probability of $p$ or an identity gate $I$ at a probability of $1-p$.
\end{definition}

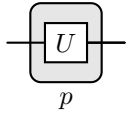
\begin{figure}
    \centering
     \begin{quantikz}
         & \probgate{U}{$p$} & \qw
     \end{quantikz}
     \caption{The circuit diagram of the probabilistic gate $U(p)$.} 
     \label{fig:example_prob_gate}
\end{figure}

A quantum circuit with at least one probabilistic gate is a probabilistic circuit.
Each probabilistic gate in the circuit is compiled independently.
\begin{definition}[Probabilistic quantum circuit]
A probabilistic circuit is a quantum circuit that contains probabilistic quantum gates.
    \label{def:probabilistic_quantum_circuit}
\end{definition}

\begin{example}
    In \cref{fig:example_prob_circ}, an example of a probabilistic quantum circuit is shown. Since $2$ probabilistic gates exist in the circuit in \cref{fig:example_prob_circ}, $4$ possible circuits as shown in \cref{fig:example_generated_circ} could be generated from compilation. 
\end{example}

Probabilistic gates introduce randomness into circuits. This randomness is brought about by the stochastic compilation process embedded in probabilistic gates.
Mid-circuit measurement provides another way of introducing randomness into quantum circuits since it enables the realization of operations relying on measurement outcomes at runtime.
\begin{remark}
    The randomness from mid-circuit measurements is introduced at runtime, while the randomness from probabilistic gates is introduced at compile time. This point is emphasized here because it implies the ideology of our optimization: Replace dynamic circuit snippet with equivalent probabilistic circuit snippet to reduce runtime overhead at only an extra static cost.
\end{remark}

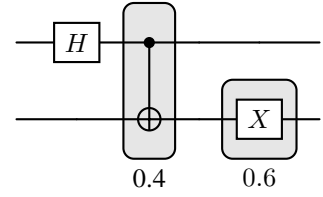
\begin{figure}
    \centering
     \begin{quantikz}
     & \gate{H} &
     \probctrl{1}{2}{0.4} & & & \\
     & & \targ{} &  & \gate{X}\gategroup[1,steps=1,style={rounded
     corners,fill=gray!20, inner
     xsep=2pt},background,label style={label
     position=below,anchor=north,yshift=-0.2cm}]{{$0.6$}} &
     \end{quantikz}
    \caption{An example of probabilistic quantum circuit $C_\text{p}$ containing a $CNOT(0.4)$, a $X(0.6)$, and a certain Hadamard gate. By compiling this probabilistic circuit, the $CNOT$ gate is generated at probability of $0.4$ and the $X$ gate is generated at probability of $0.6$.}
    \label{fig:example_prob_circ}
\end{figure}

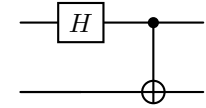
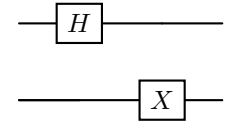
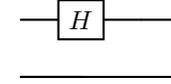
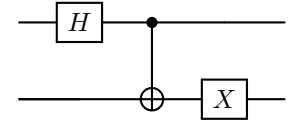
\begin{figure}
    \centering
    \begin{subfigure}{0.2\textwidth}
       \begin{quantikz}
            & \gate{H} & \ctrl{1} & \\ 
            & \qw      & \targ{} &
        \end{quantikz}
        \caption{A static circuit $C_1$, where only $CNOT$ gate appears.}
    \end{subfigure}
    \hfill
    \begin{subfigure}{0.2\textwidth}
       \begin{quantikz}
            & \gate{H} &  & \\ 
            & \qw      & \gate{X} &
        \end{quantikz}
        \caption{A static circuit $C_2$, where only $X$ gate appears.}
    \end{subfigure}
    \hfill
    \begin{subfigure}{0.2\textwidth}
       \begin{quantikz}
            & \gate{H} &  & \\ 
            & \qw      &  &
        \end{quantikz}
        \caption{A static circuit $C_3$, where neither $CNOT$ nor $X$ gate appear.}
    \end{subfigure}
    \hfill
    \begin{subfigure}{0.2\textwidth}
       \begin{quantikz}
            & \gate{H} & \ctrl{1}   && \\ 
            & \qw      & \targ{} & \gate{X} &
        \end{quantikz}
        \caption{A static circuit $C_4$, where both $CNOT$ and $X$ gate appear.}
    \end{subfigure}
        
    \caption{$4$ possible circuits generated by the probabilistic circuit in \cref{fig:example_prob_circ}.}
    \label{fig:example_generated_circ}
\end{figure}

Next, the concept of circuit ensemble is used to model the uncertainty brought about by probabilistic gates or mid-circuit measurements.

For a circuit $C$, a notation $\|C\|^{\star}$ is invented to express every possible resulting circuit of $C$.
$\|C\|^{\star}$ denotes the ensemble of all possible circuits compiled from $C$ (i.e., every probabilistic gate in $C$ is compiled) and then every mid-circuit measurement in the compiled circuit is performed. Each circuit contained in the ensemble $\|C\|^{\star}$ is paired with the probability of getting this circuit by compiling $C$ and executing every mid-circuit measurement. The definition of the notation $\|\cdot\|^{\star}$ is given in \cref{def:circuit_compiled_executed}. 

\begin{definition}
\label{def:circuit_compiled_executed}
    For a circuit $C$, the ensemble $\|C\|^{\star} \coloneqq \{ (C_i, p_i) \mid C_i$ is a static circuit achieved by compiling each probabilistic gate in $C$ and performing each mid-circuit measurement in the compiled result.$\}$
\end{definition}
\begin{remark}
    For a probabilistic circuit $C_{\text{p}}$, $\|C_{\text{p}}\|^{\star}$ is an ensemble containing all possible resulting circuits after every probabilistic gate in $C_{\text{p}}$ is compiled according to its probability; for a dynamic circuit $C_{\text{d}}$, $\|C_{\text{d}}\|^{\star}$ is an ensemble containing all possible resulting circuits after every mid-circuit measurement in $C_{\text{d}}$ is performed.
\end{remark}

\begin{example}
    For a probabilistic gate $U(p)$, 
    $\|U(p)\|^{\star} = \{(p, U), (1 - p, I) \}$, where $I$ is the identity gate.
\end{example}

The following \cref{def:multiply_ensemble} and \cref{lm:multiply_ensemble} make it possible to compose ensembles of gates together into an ensemble of circuits.
\begin{definition}
\label{def:multiply_ensemble}
     Suppose $\cdot$ is the sequential composition and $\otimes$ is the parallel composition, and $\|C_1\|^{\star}$ and $\|C_2\|^{\star}$ are two ensembles of quantum circuits. Then $\|C_1\|^{\star} \cdot \|C_2\|^{\star} = \{(p_ip_j, G^{\prime}_i\cdot G^{\prime\prime}_j) \mid \text{for each } (p_i, G^{\prime}_i) \in \|C_1\|^{\star} \text{ and } (p_j, G^{\prime\prime}_j) \in \|C_2\|^{\star} \}$, and 
     $\|C_1\|^{\star} \otimes \|C_2\|^{\star} = \{(p_ip_j, G^{\prime}_i\otimes G^{\prime\prime}_j) \mid \text{for each } (p_i, G^{\prime}_i) \in \|C_1\|^{\star} \text{ and } (p_j, G^{\prime\prime}_j) \in \|C_2\|^{\star} \}$.
\end{definition}

\begin{lemma}
\label{lm:multiply_ensemble}
    For two circuits $C_{1}$ and $C_{2}$, it holds that $\|C_{1} \cdot C_{2}\|^{\star} = \|C_{1}\|^{\star} \cdot \|C_{2}\|^{\star}$ and $\|C_{1} \otimes C_{2}\|^{\star} = \|C_{1}\|^{\star} \otimes \|C_{2}\|^{\star}$.
\end{lemma}

\begin{example}
     In the example shown in \cref{fig:example_prob_circ} and \cref{fig:example_generated_circ}, it holds by \cref{def:multiply_ensemble} and \cref{lm:multiply_ensemble} that
     $\|C_{\text{p}}\|^{\star} = \{(0.16, C_1), (0.36, C_2), (0.24, C_3), (0.24, C_4)\}$.
\end{example}

If two circuits give the same circuit ensemble, then they are equivalent after being compiled and fully executed.
This equivalence between circuits based on circuit ensemble establishes the basis for our optimization and it is formally defined in \cref{def:runtime_eq}.

\begin{definition}[Runtime-equivalence]
\label{def:runtime_eq}
    Two circuits $C_1$ and $C_2$ are runtime-equivalent, which is denoted by $C_1 \triangleq C_2$, if $\|C_1\|^{\star} = \|C_2\|^{\star}$.
    \label{def:runtime_equivalence}
\end{definition}

Then, the discussion moves to the optimization of mid-circuit measurements.
In \cref{theorem:measurement_on_pure_state}, one possible optimization is suggested where a mid-circuit measurement and its controlling gate are replaced by a rotation gate, a probabilistic $X$ gate, and a normal controlled gate. In \cref{theorem:measurement_on_pure_state_without_ctrl}, a special case of the optimization from \cref{theorem:measurement_on_pure_state}, when the mid-circuit measurement is controlling no gate, is provided. In this case, simply the mid-circuit measurement is replaced by a rotation gate and a probabilistic $X$ gate. 
\begin{theorem}
    Given a quantum state $\ket{\psi} \coloneqq \alpha\ket{0} + \beta\ket{1}$, a rotation gate $R_{\psi\rightarrow{}1}$ such that $R_{\psi\rightarrow{}1}\ket{\psi} = \ket{1}$, and a probability $p = \|\alpha\|^2$, then
        \begin{equation}
            \begin{quantikz}
            \lstick{$\ket{\psi}$}  && \meter{} \wire[d][1]{c} & \qw\\
            & \qwbundle{n}          & \gate{V}                & \qw
            \end{quantikz}
            \triangleq
            \begin{quantikz}
            \lstick{$\ket{\psi}$}  &\gate{R_{\psi\rightarrow 1}}&\gate{X}\gategroup[1,steps=1,style={rounded
             corners,fill=gray!20, inner
             xsep=2pt},background,label style={label
             position=below,anchor=north,yshift=-0.2cm}]{$p$}& \ctrl{1} & \qw\\
            & \qwbundle{n}          && \gate{V}                & \qw
            \end{quantikz}\,.
            \label{equation:theorem_measurement_on_pure_state}
        \end{equation}\
    \label{theorem:measurement_on_pure_state}
\end{theorem}
\begin{proof}
    The mid-circuit measurement in the dynamic circuit on the left-hand side of \cref{equation:theorem_measurement_on_pure_state} measures to $0$ with probability $\|\alpha\|^2$, leading to the circuit
    \begin{equation}
        C_1 \coloneqq 
        \begin{quantikz}
            \lstick{$\ket{0}$}  &   & \qw\\
            & \qwbundle{n}                 & \qw
        \end{quantikz}
    \end{equation}
    and measures to $1$ with a probability of $\|\beta\|^2$, leading to the circuit
    \begin{equation}
        C_2 \coloneqq
        \begin{quantikz}
            \lstick{$\ket{1}$}  &&&\\
            & \qwbundle{n}          & \gate{V}   & 
        \end{quantikz} \,.
    \end{equation}
    Therefore, the circuit on the left-hand side of \cref{equation:theorem_measurement_on_pure_state} can be represented by the ensemble $\{(\|\alpha\|^2, C_1), (\|\beta\|^2, C_2)\}$.
    \\
    In the probabilistic circuit on the right-hand side of \cref{equation:theorem_measurement_on_pure_state},
    at probability $p = \|\alpha\|^2$ the $X$ gate appears, and the circuit becomes
    \begin{equation}
        C_1^{\prime} \coloneqq
        \begin{quantikz}
            \lstick{$\ket{\psi}$}  &\gate{R_{\psi\rightarrow{}1}}&\gate{X}& \ctrl{1} &\\
            & \qwbundle{n}          && \gate{V}               & 
        \end{quantikz} \,,
    \end{equation}
    while at probability $1 - p =  1 - \|\alpha\|^2 = \|\beta\|^2$,
    the $X$ gate does not appear, and the circuit becomes
    \begin{equation}
        C_2^{\prime} \coloneqq
        \begin{quantikz}
            \lstick{$\ket{\psi}$}  &\gate{R_{\psi\rightarrow{}1}}& \ctrl{1} &\\
            & \qwbundle{n}          & \gate{V}               & 
        \end{quantikz} \,.
    \end{equation}
    Then the circuit on the right-hand side of \cref{equation:theorem_measurement_on_pure_state} can be represented by the ensemble $\{(\|\alpha\|^2, C_1^{\prime}), (\|\beta\|^2, C_2^{\prime})\}$.
    \\
    By definition of the rotation gate $R_{\psi\rightarrow 1}$, the circuit $C_1^{\prime}$ is equivalent to
    \begin{equation}
        \begin{quantikz}
            \lstick{$\ket{1}$}  &\gate{X}& \ctrl{1} &\\
            & \qwbundle{n}      & \gate{V}               & 
        \end{quantikz}
    \end{equation}
    which is, by definition of $X$ gate, equivalent to
    \begin{equation}
        \begin{quantikz}
            \lstick{$\ket{0}$}  && \ctrl{1} &\\
            & \qwbundle{n}      & \gate{V}               & 
        \end{quantikz}
    \end{equation}
    and since the control signal is never activated, the circuit is equivalent to
    \begin{equation}
        \begin{quantikz}
            \lstick{$\ket{0}$}   &&\\
            & \qwbundle{n}       & 
        \end{quantikz} \,.
    \end{equation}
    Therefore $C_1 = C_1^{\prime}$. 
    Using similar reasoning we get $C_2 = C_2^{\prime}$.
    Hence we get $\{(\|\alpha\|^2, C_1), (\|\beta\|^2, C_2)\} = \{(\|\alpha\|^2, C_1^{\prime}), (\|\beta\|^2, C_2^{\prime})\}$, and by \cref{def:runtime_equivalence} the two circuits in \cref{theorem:measurement_on_pure_state} are runtime-equivalent.
    
\end{proof}

The optimization step suggested by \cref{theorem:measurement_on_pure_state} is: When the input state to a mid-circuit measurement is statically determined and if it is a pure state, then the dynamic circuit snippet on the left-hand side of \cref{equation:theorem_measurement_on_pure_state} is replaced with the probabilistic circuit snippet on the right-hand side.
\begin{remark}
    The ideology behind the above optimization is to reduce runtime overhead at a static cost. Specifically, the static cost includes computing the $\ket{\psi}$ information at compile time and the extra effort needed to compile the probabilistic $X$ gate; the right-hand side of \cref{equation:theorem_measurement_on_pure_state} is lower in runtime overhead because it is free of mid-circuit measurement.
\end{remark}

A special case of this optimization step is suggested in the following \cref{theorem:measurement_on_pure_state_without_ctrl}, where the mid-circuit measurement does not control any gate. Similarly, to replace the dynamic circuit snippet on the left-hand side of \cref{equation:theorem_measurement_on_pure_state_without_ctrl}, the input state to the mid-circuit measurement needs to be statically determined and in pure state.

\begin{theorem}
    Given a quantum state $\ket{\psi} \coloneqq \alpha\ket{0} + \beta\ket{1}$, a rotation gate $R_{\psi\rightarrow{}1}$ such that $R_{\psi\rightarrow{}1}\ket{\psi} = \ket{1}$, and a probability $p = \|\alpha\|^2$, then
    \begin{equation}
        \begin{quantikz}
        \lstick{$\ket{\psi}$}  & \meter{} & \qw
        \end{quantikz}
        \triangleq
        \begin{quantikz}
        \lstick{$\ket{\psi}$}  &\gate{R_{\psi\rightarrow{}1}}&&\probgate{X}{$p$} &\qw
        \end{quantikz}\,.
        \label{equation:theorem_measurement_on_pure_state_without_ctrl}
    \end{equation}
\label{theorem:measurement_on_pure_state_without_ctrl}
\end{theorem}
\begin{proof}
    The proof is very similar to that for \cref{theorem:measurement_on_pure_state}.
\end{proof}


\paragraph{Purity test}
The optimizations on mid-circuit measurements suggested by \cref{theorem:measurement_on_pure_state} and \cref{theorem:measurement_on_pure_state_without_ctrl} both require that the input state to the mid-circuit measurement is pure.
Therefore, before applying the optimization on mid-circuit measurements, information of purity on qubits is needed.
The purity test presented in \cref{theorem:purity_test} will be performed for each of the mid-circuit measurements to decide whether it can be replaced by static components, i.e. probabilistic gates and normal quantum gates.

\begin{theorem}[Purity test]
\label{theorem:purity_test}
For a $n$-qubit quantum state $\ket{\Psi}$, the procedure in \cref{alg:purity_test} returns $true$ if its $i$-th qubit is not entangled with other qubits.

\begin{algorithm}
\label{alg:purity_test}
\caption{Purity test}
    \KwData {$i \in \{1, \dots, n\}$, $\ket{\Psi}$,  \\ 
    where $\ket{\Psi} = \alpha_{1}\ket{\psi_1} + \alpha_{2}\ket{\psi_2} + \dots + \alpha_{k}\ket{\psi_k}$,
    $1 \le k \le 2^n$; $\forall j \in \{1,\dots, k\}, \alpha_j \ne 0 $, and $\ket{\psi_1}, \dots, \ket{\psi_k}$ are $k$ different computational basis states}
    \KwResult {$b \in \mathbb{B}$}
    
    $A_0 \gets \{\alpha_j \mid j \in \{1, \dots, k\} \text{ and } \ket{\psi_j}_i = 0\}$; \hyperref[footnote_1]{\textsuperscript{1}}
    
    $A_1 \gets \{\alpha_j \mid j \in \{1, \dots, k\} \text{ and } \ket{\psi_j}_i = 1\}$; \hyperref[footnote_1]{\textsuperscript{1}}
    
    \If{$|A_0| = 0$ or $|A_1| = 0$}
        {\Return $true$;}
    
    \If{$|A_0| \ne |A_1|$}
        {\Return $false$;}
    
    $ratio \gets 0$; 
    
    \For {each $j$ such that $\alpha_j \in A_0$}{
        \If{exists $j'$ such that $\alpha_{j^{\prime}} \in A_1$
        and $\ket{\psi_j}_{\ne i} = \ket{\psi_{j^{\prime}}}_{\ne i}$\hyperref[footnote_2]{\textsuperscript{2}}}
        {
            \If{$ratio \ne 0$ and $ratio \ne \alpha_j / \alpha_{j^{\prime}}$}
            {\Return $false$;}
            {
                $ratio \gets \alpha_j / \alpha_{j^{\prime}}$;
                
                $A_1 \gets A_1 \setminus \{\alpha_{j^{\prime}}\}$;
            }

        }{
            \Return $false$;    
        }
        
    }
    \Return $true$;
\end{algorithm}
\footnotetext[1]{\label{footnote_1}For a computational basis state $\ket{\phi}$, $\ket{\phi}_p$ denotes the binary digit corresponding to the $p$-th qubit. E.g., $\ket{0010}_3 = 1, \ket{101}_2  = 0$.}
\footnotetext[2]{\label{footnote_2}For a computational basis state $\ket{\phi}$, $\ket{\phi}_{\ne p}$ denotes the binary string representation of $\ket{\phi}$ after removing the digit corresponding to the $p$-th qubit. E.g., $\ket{010101}_{\ne 3} = 01101$, $\ket{010101}_{\ne 1} = 10101$.}
\end{theorem}
For an input state containing $k$ computational basis states, the above procedure has a computational complexity of $\mathcal{O}(k^2)$.

\paragraph{Put everything together}
Finally, we are now ready to discuss the overall framework of our optimization method, which is presented in the form of pseudo-code in \cref{alg:overall_framework}.
First, QCP is applied to propagate initial constant information throughout the circuit.
Next, for each of the mid-circuit measurements in the circuit, the purity test presented in \cref{theorem:purity_test} is performed to examine whether its input state is in a pure state. 
If it is, then we can, by theorem \cref{theorem:measurement_on_pure_state} or \cref{theorem:measurement_on_pure_state_without_ctrl}, replace the mid-circuit measurement and its controlling gate with a static circuit snippet consisting of a probabilistic gate and normal quantum gates.

\paragraph{Asymptotic analysis}
For a $n$-qubit dynamic circuit of $g$ gates, 
QCP runs in $\mathcal{O}(g\cdot c^2 \cdot n)$, where $c$ is the maximum number of controls allowed for each gate.
As mentioned below \cref{theorem:purity_test}, each purity test runs in $\mathcal{O}(k^2)$ where $k$ is the number of computational basis states in the input state being tested. 
In QCP, each entanglement group is limited in $n_{max}$ many basis states \cite{chen_QCP_2023}. Suppose there are $m$ many mid-circuit measurements in the circuit, then $\mathcal{O}(m \cdot n_{max}^2)$ many steps are needed for the purity test.
Finally, the computational complexity of the optimization is $\mathcal{O}(g\cdot c^2 \cdot n + m \cdot n_{max}^2)$. Since $n_{max}$ is a constant, the optimization is polynomial in $n$, $g$, $c$ and $m$. 
\begin{algorithm}
\caption{Overall framework of optimization}
\label{alg:overall_framework}
    \KwData{A dynamic circuit $C$} 
    
    \KwResult{An optimized circuit $C_o$ with potentially less number of mid-circuit measurements} 
    
    $\text{S}_{\text{const\_info}} \gets \text{QCP}.\textbf{run}(C)$;
    
    \For {each mid-circuit measurement $M$ in C}{
        // purity-test is explained in \cref{theorem:purity_test}
        
        $purity \gets \text{purity-test}.\textbf{run}(M.\text{input\_state}, \text{S}_{\text{const\_info}})$;
        
        \If{$purity = true$}{
            // Apply \cref{theorem:measurement_on_pure_state} or \cref{theorem:measurement_on_pure_state_without_ctrl}
            
            $\text{optimize}(M)$;
        }
    }
    
\end{algorithm}

\paragraph{Demonstrating examples}

In this paragraph, our proposed optimization is applied to two dynamic circuits to briefly demonstrate its effectiveness.
In the first example depicted in \cref{fig:dem_example_1}, our proposed optimization simplified the dynamic circuit into a static circuit free of any mid-circuit measurement, significantly reducing runtime overhead. Besides, the circuit depth is decreased. 
The next example is depicted in \cref{fig:dem_example_2}. By \cref{theorem:measurement_on_pure_state}, one mid-circuit measurement, together with its controlled gate, is replaced with a rotation gate (a Hadamard gate in this case), a probabilistic $X$ gate and a controlled gate.
Another mid-circuit measurement is replaced with a normal quantum control because based on the information propagated by QCP, its input state is in one of the computational basis states. 
Although the depth of the circuit of \cref{fig:dem_example_2_simplified} increases slightly, the optimization manages to exchange dynamic components in the circuit into static ones, vastly decreasing runtime cost.

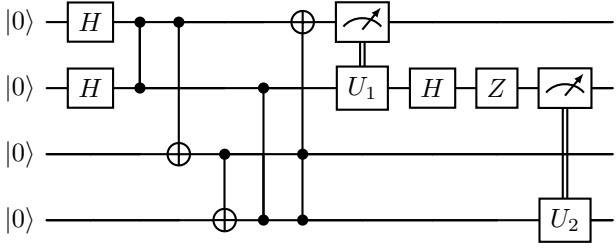
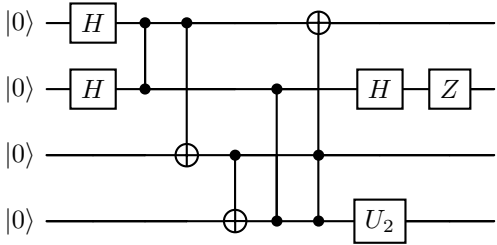
\begin{figure}
    \begin{subfigure}{1.0\linewidth}
        \begin{quantikz}[column sep=8pt, row sep={25pt,between origins}]
  \lstick{$\ket{0}$} & \gate{H} & \ctrl{1} & \ctrl{2} & \qw & \qw & \targ{} & \meter{}\wire[d][1]{c} & \qw &\qw & \qw& \qw \\
  \lstick{$\ket{0}$} & \gate{H} & \ctrl{-1} & \qw & \qw & \ctrl{2} & \qw & \gate{U_1} & \gate{H} & \gate{Z} & \meter{}\wire[d][2]{c} & \qw\\
  \lstick{$\ket{0}$} & \qw & \qw & \targ{} & \ctrl{1} & \qw & \ctrl{-2} & \qw & \qw & \qw & \qw& \qw  \\
  \lstick{$\ket{0}$} & \qw & \qw & \qw & \targ{} & \ctrl{-1} & \ctrl{-1} & \qw & \qw &\qw & \gate{U_2} & \qw\\
\end{quantikz}
        \caption{The dynamic circuit in the first demonstrating example.}
        \label{fig:dem_example_1_original}
    \end{subfigure}
    \\
    \hfill
    \\
    \begin{subfigure}{1.0\linewidth}
        
        \begin{quantikz}[column sep=9pt, row sep={25pt,between origins}]
  \lstick{$\ket{0}$} & \gate{H} & \ctrl{1} & \ctrl{2} & \qw & \qw & \targ{} & \qw &\qw & \qw \\
  \lstick{$\ket{0}$} & \gate{H} & \ctrl{-1} & \qw & \qw & \ctrl{2} & \qw & \gate{H} & \gate{Z} & \qw\\
  \lstick{$\ket{0}$} & \qw & \qw & \targ{} & \ctrl{1} & \qw & \ctrl{-2} & \qw & \qw & \qw  \\
  \lstick{$\ket{0}$} & \qw & \qw & \qw & \targ{} & \ctrl{-1} & \ctrl{-1} & \gate{U_2} &\qw & \qw\\
\end{quantikz}
        \caption{The circuit optimized by our proposed method.}
        \label{fig:dem_example_1_simplified}
    \end{subfigure}
    \caption{The first demonstrating example.}
    \label{fig:dem_example_1}
\end{figure}
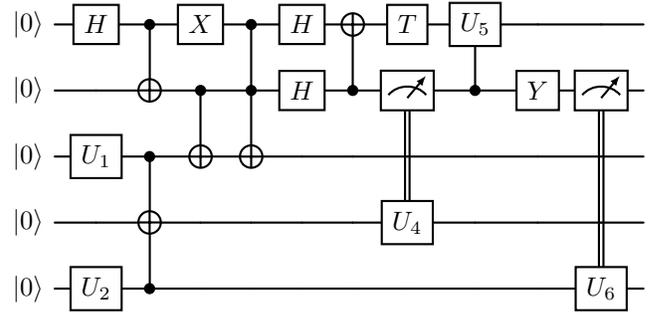
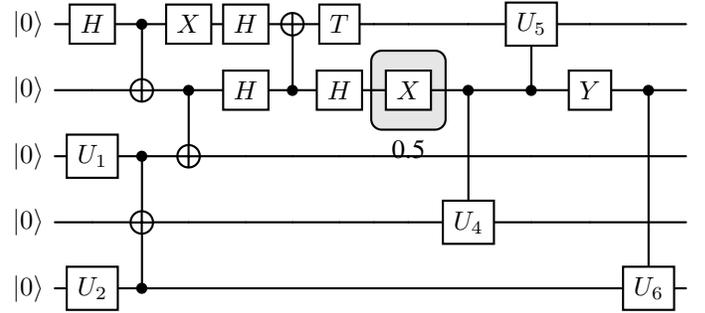
\begin{figure}
    \begin{subfigure}{1.0\linewidth}
    \begin{quantikz}[column sep=6pt, row sep={25pt,between origins}]
        \lstick{$\ket{0}$} & \gate{H}  & \ctrl{1} & \gate{X}& \ctrl{1}& \gate{H}& \targ{}  & \gate{T}              & \gate{U_5} &&&\\
        \lstick{$\ket{0}$} &           & \targ{}  & \ctrl{1}& \ctrl{1}& \gate{H}& \ctrl{-1}& \meter{}\wire[d][2]{c}& \ctrl{-1}  & \gate{Y} &\meter{}\wire[d][3]{c}&\\
        \lstick{$\ket{0}$} & \gate{U_1}& \ctrl{1} & \targ{} & \targ{} &&&&&&&          \\
        \lstick{$\ket{0}$} &           & \targ{}  &         &         &         &          & \gate{U_4} &&&&      \\
        \lstick{$\ket{0}$} & \gate{U_2}& \ctrl{-1}&         &         &         &          &           &&& \gate{U_6}&
    \end{quantikz}
        \caption{The dynamic circuit in the second demonstrating example.}
        \label{fig:dem_example_2_original}
    \end{subfigure}
    \\
    \hfill
    \\
    \begin{subfigure}{1.0\linewidth}
    \begin{quantikz}[column sep=4.5pt, row sep={25pt,between origins}]
        \lstick{$\ket{0}$} & \gate{H}  & \ctrl{1} & \gate{X}& \gate{H}& \targ{}  & \gate{T}      &&&        & \gate{U_5} &&&\\
        \lstick{$\ket{0}$} &           & \targ{}  & \ctrl{1}&  \gate{H}& \ctrl{-1}& \gate{H} && \probgate{X}{0.5} & \ctrl{2} & \ctrl{-1}  & \gate{Y} & \ctrl{3}&\\
        \lstick{$\ket{0}$} & \gate{U_1}& \ctrl{1} & \targ{} & &&&&&&&&&          \\
        \lstick{$\ket{0}$} &           & \targ{}  &         &         &     &&&     & \gate{U_4} &&&&      \\
        \lstick{$\ket{0}$} & \gate{U_2}& \ctrl{-1}&         &         &    &&&      &           &&& \gate{U_6}&
    \end{quantikz}
     \caption{The circuit optimized by our proposed method.}
        \label{fig:dem_example_2_simplified}
    \end{subfigure}
    \caption{The second demonstrating example.}
    \label{fig:dem_example_2}
\end{figure}

%% file: sections/relevant_work.tex
The technology of mid-circuit measurements is being increasingly utilized in quantum computing. 
For example, qubit-reuse, the technique of recycling inactive qubits and using them in the rest of the circuits, could use mid-circuit measurements as one of the approaches to reset inactive qubits \cite{decross_qubit-reuse_2022, brandhofer_optimal_2023, hua_exploiting_2023}.
Mid-circuit measurement also plays a vital role in Quantum Error Correction. In many Quantum Error Correction codes like Surface Code, the mid-circuit measurement is an important part of the protocol \cite{bravyi_quantum_1998, dennis_topological_2002}. 

Quantum compilation is a process that transforms high-level source code into low-level entities executable by various quantum computers sitting at the backend. In recent decades, this field has seen an outpouring of research \cite{elsharkawy_integration_2023}.
In the quantum compilation process, various optimization passes are applied to achieve certain computations with less consumption of resources \cite{Qiskit, Sivarajah_tket_2021}.

As an effort to bring static analysis to the world of quantum computing, 
Chen~\andothers~proposed the Quantum Constant Propagation, a circuit optimization pass that tries to propagate the initial information throughout the circuit and use the propagated information to simplify circuits \cite{chen_QCP_2023}.

The paradigm of probabilistic quantum circuits is similar to probabilistic programming, which encodes probability distribution in programming \cite{pfeffer_practicalPP_2016}.
While the key advantage of probabilistic programming is its flexibility in representing and manipulating uncertain knowledge, the shining point of probabilistic quantum circuits is its ability to model the uncertainty of dynamic circuits at runtime.  

%% file: sections/conclusion.tex
Mid-circuit measurement poses significant challenges for quantum computing hardware. Therefore, this paper introduces a novel optimization approach to mitigate this issue by reducing the number of mid-circuit measurements for dynamic circuits. By introducing the concept of probabilistic circuits and taking advantage of QCP, our proposed method offers a promising avenue for addressing this challenge. 
Through our contributions, this paper advances the optimization of mid-circuit measurements and offers insights into the efficiency and scalability of the dynamic circuit model.